\documentclass[11pt,a4paper]{article}
\usepackage[a4paper,total={6.00in, 8.75in}]{geometry}
\usepackage{graphicx,hyperref,enumitem}
\usepackage[round]{natbib}
\usepackage{amsmath,amssymb,xcolor}

\usepackage[normalem]{ulem}
\usepackage{amsthm}
\newtheorem{theorem}{Theorem}

\title{\textbf{A note on the $g$ and $h$ control charts}}
\author{\textsf{Chanseok Park}\\
Applied Statistics Laboratory\\ 
Department of Industrial Engineering\\ Pusan National University\\
Busan 46241, Korea
\and
\textsf{Min Wang}\\
Department of Management Science and Statistics \\
University of Texas at San Antonio \\
San Antonio, TX 78249, USA
}
\date{}

\begin{document}
%%===============================================
\maketitle
%------------------------------------------------
\begin{abstract}
%------------------------------------------------
In this note, we revisit the $g$ and $h$ control charts that are commonly used 
for monitoring the number of conforming cases between the two consecutive appearances of nonconformities.
It is known that the process parameter
of these charts is usually unknown and estimated by using the maximum likelihood
estimator and the minimum variance unbiased estimator.
However, the minimum variance unbiased estimator in the control charts has been inappropriately
used in the quality engineering literature. This observation motivates us to provide the 
correct minimum variance unbiased estimator and investigate theoretical and empirical biases of these estimators under 
consideration. Given that these charts are developed based on the underlying assumption that
samples from the process should be balanced, which is often not satisfied in many practical 
applications, we propose a method for constructing these charts with unbalanced samples.


%------------------------------------------------
\end{abstract}
%------------------------------------------------

%======================================
\section{Introduction}
%======================================
In an introductory statistics course, the geometric distribution is defined as a probability distribution 
that represents the number of failures (or normal cases) before observing the first success (or adverse case) in a
series of Bernoulli trials. Based on this distribution, \cite{Kaminsky/etc:1992} proposed 
Shewhart-type statistical control charts, the so-called $g$ and $h$ charts, 
for monitoring the number of conforming cases between the two consecutive appearances of nonconformities.
Since then they have been widely used for monitoring the control process especially in the healthcare
department; see, for example, \cite{Benneyan:1999}, \cite{Benneyan:2000}, to name just a few.  

The process parameter in the $g$ and $h$ charts is usually unknown and needs to be estimated in 
the control chart procedures. 
One can employ the maximum likelihood (ML) estimator and the minimum variance
unbiased (MVU) estimator for the process parameter. 
Of particular note is that the MVU estimator in the control charts has been inappropriately used in the quality engineering literature. 
This motivates us to obtain the correct MVU estimator 
and investigate theoretical biases of the estimators considered in this note.
Furthermore, Monte Carlo simulations are conducted to investigate the empirical biases of these estimators. 
Numerical results show that the theoretical and empirical biases of the existing estimators are severe
when the sample size is small and the value of process parameter is large and that those of the 
proposed MVU estimator are always very close to zero for all the simulated scenarios.

It deserves mentioning that these conventional $g$ and $h$ charts are developed based on 
the underlying assumption that samples from the process should be balanced 
so that the samples have the same size, whereas such an assumption can be restrictive 
and may not be satisfied in many practical applications. To overcome this issue, 
we propose the method of how to construct the $g$ and $h$ charts with
unbalanced samples.

The remainder of this note is organized as follows. In Section \ref{Section:02}, we briefly review the geometric and negative binomial distributions
and then provide the correct MVU estimator for the process parameter. In Section \ref{Section:03}, we obtain the parameter estimation with unequal sample sizes and investigate theoretical 
and empirical properties of these estimators considered in this note.
In Section \ref{Section:04}, based on the proposed estimator, 
we provide a method of constructing the $g$ and $h$ charts with unbalanced samples. 
Concluding remarks are provided in Section \ref{Section:05}. 

%======================================
\section{Basic properties of the geometric and negative binomial distributions} \label{Section:02}
%======================================
Let $Y_i$ be independent and identically distributed (iid) according to the shifted geometric distribution
with location shift $a$ and Bernoulli probability $p$ for $i=1, 2, \ldots$.
Then its probability mass function (pmf) is given by
\begin{equation} \label{EQ:pmfofGeo}
f(y) = P(Y_i=y) = p (1-p)^{y-a},
\end{equation}
where $y=a, a+1, \ldots$ and
$a$ is the known minimum possible number of events (usually $a=0,1$).
The mean and variance of $Y_i$ are respectively given by
\[
E(Y_i) = \frac{1-p}{p}+a \quad\textrm{and}\quad
\mathrm{Var}(Y_i) = \frac{1-p}{p^2}.
\]

For notational convenience, we let $T_n=\sum_{i=1}^{n}Y_i$.
Then $T_n$ has the (shifted) negative binomial with predefined location shift $na$ and Bernoulli probability $p$
and its pmf is given by
\begin{equation} \label{EQ:pmfofNGn}
g_n(t) = P(T_n=t) = \binom{t-na+n-1}{n-1} p^n (1-p)^{t-na},
\end{equation}
where $t = na, na+1, \ldots$.
The mean and variance of $T_n$ are respectively given by
\[
E(T_n) = \frac{n(1-p)}{p} + na \quad\textrm{and}\quad
\mathrm{Var}(Y_i) = \frac{n(1-p)}{p^2}.
\]

It is well known that the method of moments and the method of ML
yield the same estimator of $p$, which is given by
\begin{equation} \label{EQ:MLE0}
\hat{p}_{\mathrm{ml}} = \frac{1}{\bar{Y}-a+1},
\end{equation}
where $\bar{Y} = \sum_{i=1}^{n}Y_i/n$.
It is worth noting that this estimator is not unbiased and that we 
are able to identify the best unbiased
estimator of $p$ summarized in the following theorem. 
%------------
\begin{theorem} \label{THM:MVU}
The MVU estimator for the parameter of the geometric distribution in (\ref{EQ:pmfofGeo}) is given by
\[
\hat{p}_{\mathrm{mvu}} = \frac{n-1}{\sum_{i=1}^{n}Y_i - na + n-1}
            = \frac{(n-1)/n}{\bar{Y}-a+1-1/n}.
\]
\end{theorem}
\begin{proof}
It is immediate from \cite{Lehmann/Casella:1998} that
$T_n=\sum_{i=1}^{n}Y_i$ is a complete sufficient statistic since the joint mass functions of iid 
geometric distributions form an exponential family.
Thus, we can employ the Rao-Blackwell theorem \citep{Rao:1945,Blackwell:1947} 
to obtain the MVU estimator of $p$ as follows.

Let $\delta = {I}(Y_n = a)$ where $I(\cdot)$ is the indicator function.
Then $\delta$ is an unbiased estimator of $p$ since
\[
E(\delta) = P(Y_n = a) = p.
\]
Conditioning the unbiased estimator $\delta$ on the complete sufficient statistic $T_n=t$
and taking the expectation, we can obtain the MVU estimate, denoted by $\eta(t)$, 
due to the Rao-Blackwell theorem
\begin{equation} \label{EQ:MVU0}
\eta(t) = E\big[ {I}(Y_n = a) \mid T_n=t \big]
        =  \frac{P(Y_n=a, T_n=t)}{P(T_n=t)} .
\end{equation}
Since $Y_n=a$ and $T_n = Y_1+Y_2+\cdots+Y_n=t$, we have $T_{n-1} = Y_1+Y_2+\cdots+Y_{n-1}=t-a$ and
$T_{n-1}$ is independent of $Y_n$.  Thus, we have
\[
\eta(t) = \frac{P(Y_n=a, T_{n-1}=t-a)}{P(T_n=t)} = \frac{P(Y_n=a) \cdot P(T_{n-1}=t-a)}{P(T_n=t)}.
\]
Note that the pmfs of $Y_n=a$ and $T_{n-1}=t-a$ are given by
$f(a)$ in (\ref{EQ:pmfofGeo}) and $g_{n-1}(t-a)$ from (\ref{EQ:pmfofNGn}), respectively.
Thus, we have
\[
\eta(t) = \frac{f(a) \cdot g_{n-1}(t-a)}{g_n(t)}
= \frac{p\cdot\binom{t-a-(n-1)a+n-2}{n-2} p^{n-1} (1-p)^{t-a-(n-1)a}}{\binom{t-na+n-1}{n-1} p^n (1-p)^{t-na}},
\]
which can be simplified as
\begin{equation} \label{EQ:etat}
\eta(t) = \frac{n-1}{t- na + n-1}.
\end{equation}
Using (\ref{EQ:etat}), we obtain the MVU estimator of $p$ which is given by
\[
\hat{p}_{\mathrm{mvu}} = \frac{(n-1)/n}{\bar{Y}-a+1-1/n}.
\]
This completes the proof.
\end{proof}

%======================================
\section{Parameter estimation with unequal sample sizes} \label{Section:03}
%======================================
We assume that there are $m$ samples and each sample has different sample sizes.
We denote the size of the $i$th sample by $n_i$ for $i =1, \ldots, m$.
Let $X_{ij}$ be the number of independent Bernoulli trials (cases) until the first 
nonconforming case in the $i$th sample for $i =1, \ldots, m$ and $j = 1, \ldots, n_i$.
We assume that $X_{ij}$'s are iid geometric random variables with location shift $a$ and Bernoulli probability $p$.
Let $T_N = \sum_{i=1}^{m}\sum_{j=1}^{n_i} X_{ij}$ and $N=\sum_{i=1}^{m} n_i$.
Then it is easily seen from (\ref{EQ:pmfofNGn})
that $T_N$ has the negative binomial with predefined location shift $Na$ and Bernoulli probability $p$
and its pmf is given by 
\begin{equation} \label{EQ:pmfofNGN}
g_N(t) = P(T_N=t) = \binom{t-Na+N-1}{N-1} p^N (1-p)^{t-Na},
\end{equation}
where $t = Na, Na+1, \ldots$. It is immediate from (\ref{EQ:MLE0}) that the ML estimator with all the samples is
given by 
\[
\hat{p}_{\mathrm{ml}} = \frac{1}{\bar{\bar{X}}-a+1},
\]
where
$\bar{\bar{X}} = \sum_{i=1}^{m}\sum_{j=1}^{n_i} X_{ij} / N$.
By following Theorem~\ref{THM:MVU}, we obtain the MVU estimator of $p$ which is given by  
\[
\hat{p}_{\mathrm{mvu}} = \frac{(N-1)/N}{\bar{\bar{X}}-a+1-1/N}.
\]
To the best of our knowledge, the MVU estimator $\hat{p}_{\mathrm{mvu}}$ above 
has not yet been used in the quality engineering literature.
For example, \cite{Benneyan:2001} and \cite{MinitabSupport20} use the following estimator $\hat{p}_{\mathrm{b}}$
as the MVU estimator of $p$
%---------------------
\begin{equation} \label{EQ:pb}
\hat{p}_{\mathrm{b}} = \frac{(N-1)/N}{\bar{\bar{X}}-a+1},
\end{equation}
%---------------------
which is however not unbiased.
As an illustration, consider the case of the degenerating geometric distribution with $p=1$.
Then we have $P(X_{ij}=a) = 1$, so that $\bar{\bar{X}}=a$. 
Thus, we have $\hat{p}_{\mathrm{ml}} =1$ and $\hat{p}_{\mathrm{mvu}} =1$,
whereas $\hat{p}_{\mathrm{b}} = 1-1/N$, indicating that $\hat{p}_{\mathrm{b}}$ is not unbiased.

It is worth noting that the estimator $\hat{p}_{\mathrm{b}}$ in (\ref{EQ:pb}) was obtained
by simply multiplying the ML estimator
with the factor $(N-1)/N$, that is, $\hat{p}_{\mathrm{b}} = \hat{p}_{\mathrm{ml}} \cdot (N-1)/N$.
For the case of the exponential distribution
with the density $f(y) = \lambda e^{-\lambda y}$, which can be regarded as a continuous version of
the geometric distribution, the MVU estimator of $\lambda$ can be obtained by simply multiplying the unbiasing factor
$(N-1)/N$ with the ML estimator; see, for example, \cite{Miyakawa:1984} and \cite{Park:2010a}.
However, this technique fails to the case of the geometric distribution. 
Also, it is of interest to provide 
the inequality relation of the three estimators considered above in the following theorem.

%===================================
\begin{theorem} \label{THM:compare}
For $0<p<1$, we have
\[
\hat{p}_{\mathrm{b}} < \hat{p}_{\mathrm{mvu}} < \hat{p}_{\mathrm{ml}}.
\]
\end{theorem}
%-----------------------------------
\begin{proof}
First, we show that $\hat{p}_{\mathrm{b}} < \hat{p}_{\mathrm{mvu}}$.
Since the denominator of $\hat{p}_{\mathrm{b}}$ is always larger than that of $\hat{p}_{\mathrm{mvu}}$,
we have $\hat{p}_{\mathrm{b}} < \hat{p}_{\mathrm{mvu}}$.

Next, we show that $\hat{p}_{\mathrm{mvu}} < \hat{p}_{\mathrm{ml}}$.
To prove this, we use the fact that the mediant of the two fractions is positioned between them, that is,
\[
\frac{a}{c} < \frac{a+b}{c+d} < \frac{b}{d},
\]
where $a/c < b/d$ and $a,b,c,d>0$.
The estimator $\hat{p}_{\mathrm{ml}}$ is the mediant of $\hat{p}_{\mathrm{mvu}}$ and $(1/N)/(1/N)$, that is,
\[
\frac{1-1/N}{\bar{\bar{X}}-a+1-1/N} < \frac{1}{\bar{\bar{X}}-a+1} < \frac{1/N}{1/N},
\]
which completes the proof.
\end{proof}
%===================================
We observe from Theorem~\ref{THM:compare} that $\hat{p}_{\mathrm{b}}$ tends to underestimate the true value $p$
and that $\hat{p}_{\mathrm{ml}}$ tends to overshoot the true value.
Since $\hat{p}_{\mathrm{b}}$ and $\hat{p}_{\mathrm{ml}}$ are biased, a natural question arises: 
what are the theoretical biases of these estimators?
In what follows, we provide the first moments of these estimators so that the biases of the estimators are easily
 obtained by subtracting the true value of $p$ from their first moments. 

%===================================
\begin{theorem} \label{THM:m1}
For $0<p<1$, we have
\begin{align*}
E(\hat{p}_{\mathrm{ml}})    & = p^N \cdot {_2}F_1(N,N;N+1;1-p) 
\intertext{and}
E(\hat{p}_{\mathrm{b}})     & = \left(\frac{N-1}{N}\right) p^N \cdot {_2}F_1(N,N;N+1;1-p) ,
\end{align*}
where ${_2}F_1(\cdot)$ is the Gaussian hypergeometric function.
\end{theorem}
%-----------------------------------
\begin{proof}
If $X_i$ has the geometric distribution with location shift $a$ and $p$,
then $X_i-a$ also follows the geometric distribution with zero shift.
Without loss of generality, we may thus assume that $a=0$.
Since $\hat{p}_{\mathrm{ml}} = 1/(\bar{\bar{X}}+1) = N/(T_N+N)$,
it is immediate upon using (\ref{EQ:pmfofNGN}) that we have
\[
E(\hat{p}_{\mathrm{ml}})
= \sum_{t=0}^{\infty} \frac{N}{t+N} \cdot g_N(t)
= \sum_{t=0}^{\infty} \frac{N}{t+N} \cdot \binom{t+N-1}{N-1} p^N (1-p)^{t},
\]
that is,
\[
E(\hat{p}_{\mathrm{ml}})
= \frac{Np^N}{(1-p)^N} \sum_{t=0}^{\infty} \binom{t+N-1}{N-1} \frac{(1-p)^{t+N}}{t+N}.
\]
Using the identity $(1-p)^{t+N}/(t+N) = \int_p^1 (1-y)^{t+N-1} dy$, we have
\begin{align}
E(\hat{p}_{\mathrm{ml}})
&= \frac{Np^N}{(1-p)^N} \sum_{t=0}^{\infty} \binom{t+N-1}{N-1} \int_p^1 (1-y)^{t+N-1} dy  \notag \\
&= \frac{Np^N}{(1-p)^N} \int_p^1 \frac{(1-y)^{N-1}}{y^N}
    \left[ \sum_{t=0}^{\infty} \binom{t+N-1}{N-1} y^N (1-y)^{t}\right] dy.
\label{EQ:phatml1}
\end{align}
Since $\binom{t+N-1}{N-1} y^N (1-y)^{t}$ is the pmf of the negative binomial distribution,
we have 
\[
\sum_{t=0}^{\infty} \binom{t+N-1}{N-1} y^N (1-y)^{t}=1.
\]
Thus, Equation (\ref{EQ:phatml1}) can be further simplified as 
\[
E(\hat{p}_{\mathrm{ml}})
= \frac{Np^N}{(1-p)^N} \int_p^1 {(1-y)^{N-1}}{y^{-N}}  dy.
\]
Using the integration by substitution with $x=1-y$, the above is written as
\begin{align*}
E(\hat{p}_{\mathrm{ml}})
&= \frac{Np^N}{(1-p)^N} \int_0^{1-p} x^{N-1} (1-x)^{-N} dx \\
&= \frac{Np^N}{(1-p)^N} \cdot B_{1-p}(N,1-N),
\end{align*}
where $B_x(a,b)$ is the  incomplete beta function defined as
\[
B_x(a,b) = \int_0^x y^{a-1} (1-y)^{b-1} dy.
\]

It deserves mentioning that the calculation of $B_{1-p}(N,1-N)$ can be complex
because few software packages provide its calculation with negative argument.
To deal with this difficulty, one can use the hypergeometric representation of the incomplete beta function
\citep{Dutka:1981,Ozarslan/Ustaoglu:2019} which is given by 
\begin{equation} \label{EQ:BetaHyper}
B_x(a,b) = \frac{x^a}{a} \cdot {_2}F_1(a,1-b;a+1;x).
\end{equation}
Here ${_p}F_q(\cdot)$ is the hypergeometric function \citep{Abramowitz:1964,Seaborn:1991} and it is defined as
\begin{equation} \label{EQ:defhyper}
{_p}F_q(a_1,\ldots,a_p; b_1,\ldots,b_q; z)
 = \sum_{n=0}^{\infty} \frac{(a_1)_n\cdots (a_p)_n}{ (b_1)_n\cdots (b_q)_n} \frac{z^n}{n!},
\end{equation}
where $(a)_n$ is the Pochhammer symbol for the rising factorial
defined as $(a)_0=1$ and $(a)_n=a(a+1)\cdots(a+n-1)$ for $n=1,2,\ldots$.
%----------
Thus, by using (\ref{EQ:BetaHyper}), we have
\begin{equation} \label{EQ:Ephatml}
E(\hat{p}_{\mathrm{ml}})
= p^N \cdot {_2}F_1(N,N;N+1;1-p).
\end{equation}
Note that we can easily obtain $E(\hat{p}_{\mathrm{b}})$ since $\hat{p}_{\mathrm{b}} = \hat{p}_{\mathrm{ml}} \cdot (N-1)/N$.
This completes the proof.
\end{proof}
%==========

By using the well-known Euler transformation formula for the hypergeometric function 
\citep{Miller:2011} which is given by 
\[
{_2}F_1(a,b;c;z) = (1-z)^{c-a-b} {_2}F_1(c-a,c-b;c;z), 
\]
we obtain $E(\hat{p}_{\mathrm{ml}}) = p \cdot {_2}F_1(1,1;N+1;1-p)$. 
Then according the definition of the hypergeometric function in (\ref{EQ:defhyper}), we have
\begin{align*}
E(\hat{p}_{\mathrm{ml}}) 
&= p \sum_{n=0}^{\infty} \frac{(1)_n (1)_n}{(N+1)_k} \frac{(1-p)^n}{n!}  \\
&= p \sum_{n=0}^{\infty} \frac{N!~n!}{(N+n)!} (1-p)^n \\
&= p + \sum_{n=1}^{\infty}  \frac{p(1-p)^n}{\binom{N+n}{n}} 
\end{align*}
since $(1)_n = n!$ and $(N+1)_n=(N+n)! / N!$. 
%\[
%E(\hat{p}_{\mathrm{ml}}) = p \sum_{n=0}^{\infty} \frac{N!~n!}{(N+n)!} (1-p)^n
% = p + \sum_{n=1}^{\infty}  \frac{p(1-p)^n}{\binom{N+n}{n}}.
%\]
Then the biases of the estimators $\hat{p}_{\mathrm{ml}}$ and $\hat{p}_{\mathrm{b}}$ are obtained as
\begin{align*}
\mathrm{Bias}(\hat{p}_{\mathrm{ml}}) &=  \sum_{n=1}^{\infty}  \frac{p(1-p)^n}{\binom{N+n}{n}} \\
\intertext{and}
\mathrm{Bias}(\hat{p}_{\mathrm{b}})  &= -\frac{p}{N} + \frac{N-1}{N}\sum_{n=1}^{\infty}  \frac{p(1-p)^n}{\binom{N+n}{n}},
\end{align*}
respectively. It should be noted that the R language
provides the \texttt{hypergeo} package to calculate the hypergeometric function; see \cite{Hankin:2016}.
We can calculate the theoretical values of the biases and provide these values in Figure~\ref{FIG:BiasMSE} 
along with the empirical values.
It deserves mentioning that the theoretical bias of $\hat{p}_{\mathrm{mvu}}$ is trivially zero.

In what follows, we provide the second moments of the estimators so that 
their variances can be easily obtained using them.

%===================================
\begin{theorem} \label{THM:m2}
For $0<p<1$, we have
\begin{align*}
E(\hat{p}_{\mathrm{ml}}^2)    & = p^N \cdot {_3}F_2(N,N,N; N+1,N+1; 1-p), \\
E(\hat{p}_{\mathrm{b}}^2)     & = \frac{(N-1)^2 p^N}{N^2} \cdot {_3}F_2(N,N,N; N+1,N+1; 1-p), 
\intertext{and}
E(\hat{p}_{\mathrm{mvu}}^2)   & = p^N \cdot {_2}F_1(N-1, N-1; N; 1-p).
\end{align*}
\end{theorem}
%-----------------------------------
\begin{proof}
We first note that 
\begin{align*}
E(\hat{p}_{\mathrm{ml}}^2)
&= \sum_{t=0}^{\infty} \left(\frac{N}{t+N}\right)^2 \cdot g_N(t)  \\
&= \sum_{t=0}^{\infty} \left(\frac{N}{t+N}\right)^2 \cdot \binom{t+N-1}{N-1} p^N (1-p)^{t}  \\
&= \frac{Np^N}{(1-p)^N} \sum_{t=0}^{\infty} \frac{N}{t+N} \cdot \binom{t+N-1}{N-1} \frac{(1-p)^{t+N}}{t+N}.
\end{align*}
By using the identity $(1-p)^{t+N}/(t+N) = \int_p^1 (1-y)^{t+N-1} dy$, we have
\begin{align*}
E(\hat{p}_{\mathrm{ml}}^2)
&= \frac{Np^N}{(1-p)^N} \sum_{t=0}^{\infty} \frac{N}{t+N} \cdot \binom{t+N-1}{N-1} \int_p^1 (1-y)^{t+N-1} dy \\
&= \frac{Np^N}{(1-p)^N} \int_p^1 \frac{(1-y)^{N-1}}{y^N}
   \left[\sum_{t=0}^{\infty} \frac{N}{t+N} \binom{t+N-1}{N-1} y^N (1-y)^t\right] dy.
\end{align*}
The term in the integrand, $\sum_{t=0}^{\infty} \frac{N}{t+N} \binom{t+N-1}{N-1} y^N (1-y)^t$, is essentially
the same as the first moment of $\hat{p}_{\mathrm{ml}}$ with probability $y$.
Thus, it follows from (\ref{EQ:Ephatml}) that 
\[
\sum_{t=0}^{\infty} \frac{N}{t+N} \binom{t+N-1}{N-1} y^N (1-y)^t = y^N \cdot {_2}F_1(N,N;N+1;1-y),
\]
which results in
\begin{align}
E(\hat{p}_{\mathrm{ml}}^2)
&= \frac{Np^N}{(1-p)^N} \int_p^1  (1-y)^{N-1}\cdot {_2}F_1(N,N;N+1;1-y) dy  \notag \\
&= \frac{Np^N}{(1-p)^N} \int_0^{1-p} x^{N-1}\cdot {_2}F_1(N,N;N+1;x) dx .
\label{THM:EQ:Ephat2}
\end{align}
Using the general integral representation for ${_{p+k}}F_{q+k}$ in Theorem~38 of \cite{Rainville:1960}
and Section~2 of \cite{Driver/Johnston:2006}, we have
\begin{equation} \label{THM:EQ:3F2}
\int_0^{1-p} x^{N-1} \cdot {_2}F_1(N,N;N+1;x) dx
= \frac{(1-p)^{N}}{N} \cdot {_3}F_2(N,N,N;N+1,N+1;1-p).
\end{equation}
Substituting (\ref{THM:EQ:3F2}) into (\ref{THM:EQ:Ephat2}), we obtain the first result.
The second result is easily obtained from $\hat{p}_{\mathrm{b}} = \hat{p}_{\mathrm{ml}} \cdot (N-1)/N$.

Next, we have
\begin{align*}
E(\hat{p}_{\mathrm{mvu}}^2)
&= \sum_{t=0}^{\infty} \left(\frac{N-1}{t+N-1}\right)^2 \cdot g_N(t) \\
&= \sum_{t=0}^{\infty} \left(\frac{N-1}{t+N-1}\right)^2 \cdot \binom{t+N-1}{N-1} p^N (1-p)^{t}  \\
&= \sum_{t=0}^{\infty} \left(\frac{N-1}{t+N-1}\right) \cdot \binom{t+N-2}{N-2} p^N (1-p)^{t}  \\
&= \frac{(N-1)p^N}{(1-p)^{N-1}} \sum_{t=0}^{\infty} \cdot \binom{t+N-2}{N-2} \frac{(1-p)^{t+N-1}}{t+N-1}.
\end{align*}
Using the identity $(1-p)^{t+N-1}/(t+N-1) = \int_p^1 (1-y)^{t+N-2} dy$, we have
\begin{align*}
E(\hat{p}_{\mathrm{mvu}}^2)
&= \frac{(N-1)p^N}{(1-p)^{N-1}} \int_p^1 \frac{(1-y)^{N-2}}{y^{N-1}}
   \left[ \sum_{t=0}^{\infty} \binom{t+N-2}{N-2} (1-y)^t y^{N-1}\right] dy \\
&= \frac{(N-1)p^N}{(1-p)^{N-1}} \int_p^1 \frac{(1-y)^{N-2}}{y^{N-1}} dy  \\
&= \frac{(N-1)p^N}{(1-p)^{N-1}} \int_0^{1-p} x^{N-2} (1-x)^{-(N-1)} dx  \\
&= \frac{(N-1)p^N}{(1-p)^{N-1}} \cdot B_{1-p} (N-1, -N+2).
\end{align*}
Then it is immediate upon using the hypergeometric representation of the incomplete beta function 
in (\ref{EQ:BetaHyper}) that we have the result, which completes the proof.
\end{proof}
%-----------------------------------

It should be noted that based on the Euler transformation formula for the hypergeometric function,
we can rewrite
\begin{align*}
E(\hat{p}_{\mathrm{mvu}}^2) 
&= p^2 \cdot {_2}F_1(1,1;N;1-p)  \\
&= p^2 + \frac{\sum_{n=1}^N p^2 (1-p)^n}{\binom{N-1+n}{n}}, 
\end{align*}
which results in 
\[
\mathrm{Var}(\hat{p}_{\mathrm{mvu}}) = \frac{\sum_{n=1}^N p^2 (1-p)^n}{\binom{N-1+n}{n}}. 
\]

In addition, we also conduct Monte Carlo simulations to study empirical biases of these estimators under consideration.
For each simulation, we generate $(n_1,n_2)=(1,1)$, $(2,3)$, $(5,5)$, $(10,10)$ samples from the geometric distribution
with Bernoulli probability $p=0.1$, $0.3$, $0.5$, $0.7$, $0.9$ with the location shift $a$ being always zero.
To obtain empirical biases and empirical mean square errors (MSEs),
we iterate this experiment $I=10,000$ times.
It should be noted that the existing methods are all biased so that it is more appropriate to compare their empirical MSEs
instead of the empirical variances.  
The empirical biases and MSEs are provided in Tables~\ref{TBL:bias} and \ref{TBL:MSE}. 
The values of the theoretical MSEs are easily obtained using Theorems~\ref{THM:m1} and \ref{THM:m2}
and we also plot the these values along with the biases in Figure~\ref{FIG:BiasMSE}.
In the figure, to compare the empirical and theoretical values, we also superimposed the empirical values 
with the legends {\Large$\circ$} ($n_1=1$, $n_2=1$), $\times$ ($n_1=2$, $n_2=3$), and $\bullet$ ($n_1=5$, $n_2=5$).

%%---------------------------------------------------
\begin{figure}[t!]
\centering%
\includegraphics{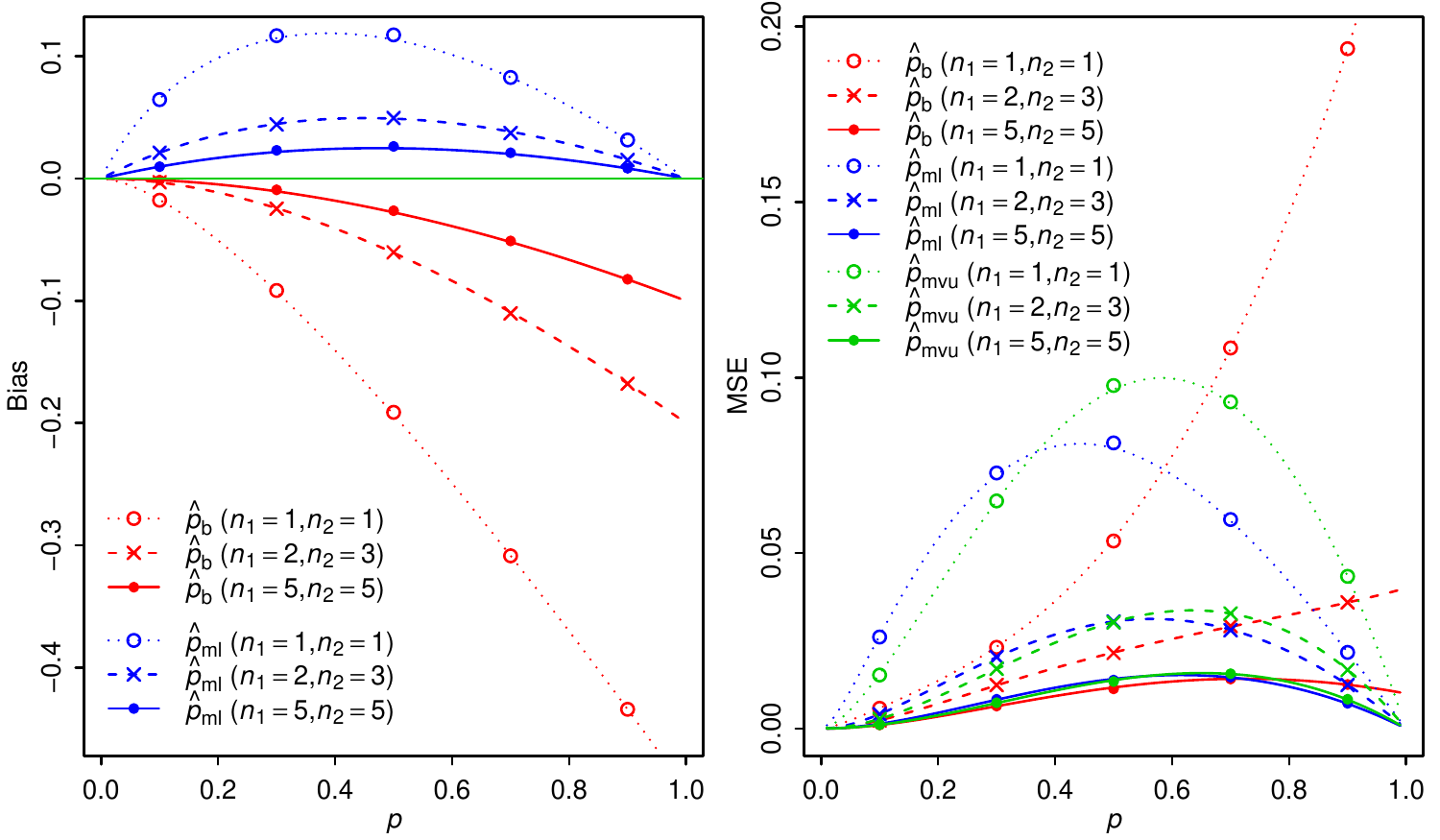}  %%-Figure
\caption{Theoretical values of the biases and MSEs of the estimators under consideration.
The empirical values are denoted by the legends {\Large$\circ$}, $\times$, and $\bullet$.}
\label{FIG:BiasMSE}
\end{figure}
%%---------------------------------------------------

%-------------------------------------------------------------------------
\begin{table}[t!]
\caption{\label{TBL:bias}
Empirical biases of $\hat{p}_{\mathrm{b}}$, $\hat{p}_{\mathrm{mvu}}$ and $\hat{p}_{\mathrm{ml}}$.}
\medskip
\centering%
\renewcommand{\arraystretch}{0.75} % this reduces the vertical spacing between rows
\begin{small}
\begin{tabular}{cllrlrlrlr} %% \\[-0.5ex]
\hline
%---------------------------------------------------------------------------
\multicolumn{2}{r}{$(n_1, n_2)$} && $(1,1)$     && $(2,3)$     && $(5,5)$ && $(10,10)$ \\
%---------------------------------------------------------------------------
\hline
$p=0.1$ & $\hat{p}_{\mathrm{b}}$ &&  $-0.01768$ && $-0.00301$ && $-0.00122$ && $-0.00061$ \\
    & $\hat{p}_{\mathrm{mvu}}$   &&  $-0.00060$ && $ 0.00005$ && $ 0.00000$ && $-0.00006$ \\
    &   $\hat{p}_{\mathrm{ml}}$  &&  $ 0.06464$ && $ 0.02124$ && $ 0.00975$ && $ 0.00462$ \\
%-------------------------------------------------------------------------
\cline{1-2}  \cline{4-4} \cline{6-6} \cline{8-8} \cline{10-10}
$p=0.3$ & $\hat{p}_{\mathrm{b}}$ &&  $-0.09153$ && $-0.02461$ && $-0.00917$ && $-0.00492$ \\
    & $\hat{p}_{\mathrm{mvu}}$   &&  $ 0.00176$ && $-0.00044$ && $ 0.00133$ && $-0.00008$ \\
    &   $\hat{p}_{\mathrm{ml}}$  &&  $ 0.11693$ && $ 0.04424$ && $ 0.02314$ && $ 0.01062$ \\
%-------------------------------------------------------------------------
\cline{1-2}  \cline{4-4} \cline{6-6} \cline{8-8} \cline{10-10}
$p=0.5$ & $\hat{p}_{\mathrm{b}}$ &&  $-0.19120$ && $-0.06032$ && $-0.02628$ && $-0.01309$ \\
    & $\hat{p}_{\mathrm{mvu}}$   &&  $ 0.00481$ && $ 0.00051$ && $ 0.00144$ && $ 0.00004$ \\
    &   $\hat{p}_{\mathrm{ml}}$  &&  $ 0.11759$ && $ 0.04960$ && $ 0.02636$ && $ 0.01254$ \\
%-------------------------------------------------------------------------
\cline{1-2}  \cline{4-4} \cline{6-6} \cline{8-8} \cline{10-10}
$p=0.7$ & $\hat{p}_{\mathrm{b}}$ &&  $-0.30864$ && $-0.11016$ && $-0.05098$ && $-0.02629$ \\
    & $\hat{p}_{\mathrm{mvu}}$   &&  $ 0.00017$ && $-0.00120$ && $ 0.00107$ && $-0.00114$ \\
    &   $\hat{p}_{\mathrm{ml}}$  &&  $ 0.08272$ && $ 0.03731$ && $ 0.02114$ && $ 0.00917$ \\
%-------------------------------------------------------------------------
\cline{1-2}  \cline{4-4} \cline{6-6} \cline{8-8} \cline{10-10}
$p=0.9$ & $\hat{p}_{\mathrm{b}}$ &&  $-0.43420$ && $-0.16783$ && $-0.08233$ && $-0.04134$ \\
    & $\hat{p}_{\mathrm{mvu}}$   &&  $-0.00005$ && $-0.00030$ && $ 0.00023$ && $-0.00049$ \\
    &   $\hat{p}_{\mathrm{ml}}$  &&  $ 0.03160$ && $ 0.01521$ && $ 0.00852$ && $ 0.00385$ \\
%-------------------------------------------------------------------------
\hline
\end{tabular}
\end{small}
\end{table}
%---------------------------

%---------------------------
\begin{table}[t!]
\caption{\label{TBL:MSE}
Empirical MSEs of $\hat{p}_{\mathrm{b}}$, $\hat{p}_{\mathrm{mvu}}$ and $\hat{p}_{\mathrm{ml}}$.}
\medskip
\centering%
\renewcommand{\arraystretch}{0.75} % this reduces the vertical spacing between rows
\begin{small}
\begin{tabular}{cllrlrlrlr} %% \\[-0.5ex]
\hline
%---------------------------------------------------------------------------
\multicolumn{2}{r}{$(n_1, n_2)$} && $(1,1)$     && $(2,3)$     && $(5,5)$ && $(10,10)$ \\
%---------------------------------------------------------------------------
\hline
$p=0.1$ & $\hat{p}_{\mathrm{b}}$ && $0.00579$ && $0.00236$ && $0.00105$ && $ 0.00050$   \\
    & $\hat{p}_{\mathrm{mvu}}$   && $0.01527$ && $0.00274$ && $0.00111$ && $ 0.00052$   \\
    &   $\hat{p}_{\mathrm{ml}}$  && $0.02609$ && $0.00412$ && $0.00140$ && $ 0.00058$   \\
%-------------------------------------------------------------------------
\cline{1-2}  \cline{4-4} \cline{6-6} \cline{8-8} \cline{10-10}
$p=0.3$ & $\hat{p}_{\mathrm{b}}$ && $0.02317$ && $0.01242$ && $0.00642$ && $ 0.00320$   \\
    & $\hat{p}_{\mathrm{mvu}}$   && $0.06487$ && $0.01706$ && $0.00737$ && $ 0.00340$   \\
    &   $\hat{p}_{\mathrm{ml}}$  && $0.07284$ && $0.02041$ && $0.00836$ && $ 0.00363$   \\
%-------------------------------------------------------------------------
\cline{1-2}  \cline{4-4} \cline{6-6} \cline{8-8} \cline{10-10}
$p=0.5$ & $\hat{p}_{\mathrm{b}}$ && $0.05345$ && $0.02156$ && $0.01133$ && $ 0.00588$   \\
    & $\hat{p}_{\mathrm{mvu}}$   && $0.09775$ && $0.03027$ && $0.01341$ && $ 0.00636$   \\
    &   $\hat{p}_{\mathrm{ml}}$  && $0.08140$ && $0.03047$ && $0.01383$ && $ 0.00649$   \\
%-------------------------------------------------------------------------
\cline{1-2}  \cline{4-4} \cline{6-6} \cline{8-8} \cline{10-10}
$p=0.7$ & $\hat{p}_{\mathrm{b}}$ && $0.10844$ && $0.02922$ && $0.01413$ && $ 0.00724$   \\
    & $\hat{p}_{\mathrm{mvu}}$   && $0.09309$ && $0.03281$ && $0.01564$ && $ 0.00757$   \\
    &   $\hat{p}_{\mathrm{ml}}$  && $0.05957$ && $0.02808$ && $0.01468$ && $ 0.00734$   \\
%-------------------------------------------------------------------------
\cline{1-2}  \cline{4-4} \cline{6-6} \cline{8-8} \cline{10-10}
$p=0.9$ & $\hat{p}_{\mathrm{b}}$ && $0.19371$ && $0.03597$ && $0.01256$ && $ 0.00512$   \\
    & $\hat{p}_{\mathrm{mvu}}$   && $0.04339$ && $0.01671$ && $0.00837$ && $ 0.00409$   \\
    &   $\hat{p}_{\mathrm{ml}}$  && $0.02172$ && $0.01242$ && $0.00721$ && $ 0.00379$   \\
%-------------------------------------------------------------------------
\hline
\end{tabular}
\end{small}
\end{table}
%---------------------------

The values of the empirical biases of $\hat{p}_{\mathrm{b}}$ 
are always negative and those of $\hat{p}_{\mathrm{ml}}$ are always positive, 
which is expected from Theorem~\ref{THM:compare}, 
and both biases tend to decrease as the sample sizes increase.
It is worth noting that the bias of $\hat{p}_{\mathrm{b}}$ is really serious, 
especially when the sample size is small and the probability $p$ is large. However,
the empirical biases of $\hat{p}_{\mathrm{mvu}}$ are very close to zero for all the cases
as expected from the fact that its theoretical bias is zero.
Numerical results clearly show that the proposed estimator $\hat{p}_{\mathrm{mvu}}$ outperforms
the existing estimators.
On the other hand, the bias of $\hat{p}_{\mathrm{ml}}$ is larger when $p$ is around 0.5.
With $n_1=1$ and $n_2=1$, the bias of $\hat{p}_{\mathrm{b}}$ can reach around 0.5 with $p$ close to 1
and that of $\hat{p}_{\mathrm{ml}}$ can reach around 0.1 with $p$ around 0.5.
Considering that the value of $p$ is always in $(0,1)$,
the biases of $\hat{p}_{\mathrm{b}}$ and $\hat{p}_{\mathrm{ml}}$
are really serious. As $N$ gets larger, the bias gets smaller, 
whereas the bias of $\hat{p}_{\mathrm{b}}$ is still severe with a large value of $p$.

%======================================
\section{Construction of the $g$ and $h$ control charts}  \label{Section:04}
%======================================
As we did earlier, we let $X_{ij}$ be the number of independent Bernoulli trials (cases)
until the first nonconforming case in the $i$th sample for $i=1,2,\ldots,m$ and $j=1,2,\ldots,n_i$.
Then $X_{ij}$'s are iid geometric random variables with location shift $a$ and $p$.
Let $\bar{X}_k$ be the mean of the $k$th sample with sample size $n_k$.

Based on the asymptotic theory, we have
\[
\frac{\bar{X}_k - \mu}{\sqrt{\sigma^2 / n_k}} \stackrel{\bullet}{\sim} N(0,1),
\]
where $\mu=E(X_{kj}) = (1-p)/p+a$ and $\sigma^2=\mathrm{Var}(X_{kj}) = (1-p)/p^2$.
We can construct the control chart for average number of events per subgroup (the $h$ chart)
 with $\mathrm{CL} \pm g\cdot\mathrm{SE}$ control limits
\[
\frac{\bar{X}_k - \mu_k}{\sqrt{\sigma^2 / n_k}} = \pm g,
\]
which results in the upper control limit (UCL), lower control limit (LCL) and center line (CL) as follows
\begin{align}
\mathrm{UCL} &= {\mu} + g \sqrt{\frac{\sigma^2}{n_k}} = \frac{1-p}{p}+a  + g\sqrt{\frac{1-p}{n_k p^2}}, \notag \\
\mathrm{CL}  &= {\mu}  = \frac{1-p}{p}+a,   \label{EQ:Limits-h-chart}  \\
\mathrm{LCL} &= {\mu} - g \sqrt{\frac{\sigma^2}{n_k}} = \frac{1-p}{p}+a  - g\sqrt{\frac{1-p}{n_k p^2}} . \notag
\end{align}
It deserves mentioning that the American Standard uses $g=3$ with an ideal false alarm rate 0.27\%
and British Standard uses $g=3.09$ with 0.20\%.

By setting up $(n_k \bar{X}_k - n_k \mu_k)/\sqrt{n_k \sigma^2} = \pm g$, we can also construct the control chart
for the total number of events per subgroup (the $g$ chart) and its control limits are given by
\begin{align}
\mathrm{UCL} &= n_k{\mu} + g \sqrt{n_k \sigma^2} = n_k\left(\frac{1-p}{p}+a\right) + g\sqrt{\frac{n_k(1-p)}{p^2}}, \notag \\
\mathrm{CL}  &= n_k{\mu}  = n_k\left(\frac{1-p}{p}+a\right),    \label{EQ:Limits-g-chart}   \\
\mathrm{LCL} &= n_k{\mu} - g \sqrt{n_k \sigma^2} = n_k\left(\frac{1-p}{p}+a\right) - g\sqrt{\frac{n_k(1-p)}{p^2}} . \notag
\end{align}

%_------------------------------------
In practice, the parameters $\mu$ and $\sigma^2$ are unknown and can be estimated by substituting
an estimator of $p$ through the relationship $\mu = 1/p-1+a$ and $\sigma^2 = (1-p)/p^2$. 
However, a care should be taken in this case.
For example, $\hat{p}_{\mathrm{mvu}}$ is unbiased for $p$, but $1/\hat{p}_{\mathrm{mvu}}$ is not unbiased for $1/p$.
We have shown that $\hat{p}_{\mathrm{ml}}$ is not unbiased for $p$,
whereas $1/\hat{p}_{\mathrm{ml}}$ is actually unbiased for $1/p$.
Thus, we estimate $\mu=1/p-1+a$ using $\hat{\mu} = 1/\hat{p}_{\mathrm{ml}}-1+a$, which results in
$\hat{\mu} = \bar{\bar{X}}$.
Since $\bar{\bar{X}}=T_N/N = \sum_{i=1}^{m}\sum_{j=1}^{n_i} X_{ij}/N$ is a complete sufficient statistic,
$\hat{\mu} = \bar{\bar{X}}$ is the MVU estimator of $\mu$ due to the Lehmann-Scheff\'{e} theorem.
For more details on this theorem, see Theorem 7.4.1 of \cite{Hogg/McKean/Craig:2013}.
It should be noted that $\hat{\mu} = \bar{\bar{X}}$ is also the ML estimator because of
the invariance property of the ML estimator \citep[for example, see Theorem 7.2.10 of][]{Casella/Berger:2002}.
Thus, it is clear that one should use $\hat{\mu} = \bar{\bar{X}}$ to estimate the CL, which results in
$\mathrm{CL} = \bar{\bar{X}}$ ($h$ chart) and $\mathrm{CL} =n_k \bar{\bar{X}}$ ($g$ chart).

To estimate $\sigma^2$, we consider the ML estimator of $\sigma^2$ by
plugging $\hat{p}_{\mathrm{ml}}$ into $\sigma^2 = (1-p)/p^2$, which results in 
\begin{equation} \label{EQ:sigma2ml}
\hat{\sigma}^2_{\mathrm{ml}} = (\bar{\bar{X}}-a) (\bar{\bar{X}}-a+1).
\end{equation}
The MVU estimator of $\sigma^2$ is also easily obtained using
the Lehmann-Scheff\'{e} theorem with $E\big[(T_N/N)\cdot(T_N+N)/(N+1)\big] = (1-p)/p^2$.
Then we have
\begin{equation} \label{EQ:sigma2mvu}
\hat{\sigma}^2_{\mathrm{mvu}} = \frac{N}{N+1}  (\bar{\bar{X}}-a) (\bar{\bar{X}}-a+1).
\end{equation}

Using $\hat{\mu} = \bar{\bar{X}}$ and $\hat{\sigma}^2_{\mathrm{ml}}$ in (\ref{EQ:sigma2ml}) along with
 (\ref{EQ:Limits-h-chart}) and (\ref{EQ:Limits-g-chart}), we can construct the ML-based $h$ and $g$ charts as follows.
%-----------
\begin{itemize}
\item $h$ chart:
\begin{align*}
\mathrm{UCL} &= \bar{\bar{X}} + g\sqrt{\frac{(\bar{\bar{X}}-a) (\bar{\bar{X}}-a+1)}{n_k}},  \\
\mathrm{CL}  &= \bar{\bar{X}},   \\
\mathrm{LCL} &= \bar{\bar{X}} - g\sqrt{\frac{(\bar{\bar{X}}-a) (\bar{\bar{X}}-a+1)}{n_k}}.
\end{align*}
\item $g$ chart:
\begin{align*}
\mathrm{UCL} &= n_k\bar{\bar{X}} + g \sqrt{n_k (\bar{\bar{X}}-a) (\bar{\bar{X}}-a+1)},  \\
\mathrm{CL}  &= n_k\bar{\bar{X}},   \\
\mathrm{LCL} &= n_k\bar{\bar{X}} - g \sqrt{n_k (\bar{\bar{X}}-a) (\bar{\bar{X}}-a+1)}.
\end{align*}
\end{itemize}
%-----------
Also, using $\hat{\mu} = \bar{\bar{X}}$ and $\hat{\sigma}^2_{\mathrm{mvu}}$ in (\ref{EQ:sigma2mvu}) along with
 (\ref{EQ:Limits-h-chart}) and (\ref{EQ:Limits-g-chart}), we can construct the MVU-based $h$ and $g$ charts as follows.
%-----------
\begin{itemize}
\item $h$ chart:
\begin{align*}
\mathrm{UCL} &= \bar{\bar{X}} + g\sqrt{\frac{N}{N+1}\frac{(\bar{\bar{X}}-a) (\bar{\bar{X}}-a+1)}{n_k}},  \\
\mathrm{CL}  &= \bar{\bar{X}},   \\
\mathrm{LCL} &= \bar{\bar{X}} - g\sqrt{\frac{N}{N+1}\frac{(\bar{\bar{X}}-a) (\bar{\bar{X}}-a+1)}{n_k}}.
\end{align*}
\item $g$ chart:
\begin{align*}
\mathrm{UCL} &= n_k\bar{\bar{X}} + g \sqrt{\frac{n_k N}{N+1} (\bar{\bar{X}}-a) (\bar{\bar{X}}-a+1)},  \\
\mathrm{CL}  &= n_k\bar{\bar{X}},   \\
\mathrm{LCL} &= n_k\bar{\bar{X}} - g \sqrt{\frac{n_k N}{N+1} (\bar{\bar{X}}-a) (\bar{\bar{X}}-a+1)}.
\end{align*}
\end{itemize}
%-----------

It should be noted that \cite{Kaminsky/etc:1992} provide the control limits for the MVU-based $h$ and $g$ charts in
their Table~1, but these limits are based on $\hat{p}_{\mathrm{b}}$ which is not the MVU.
Also, one can also construct the control limits by plugging the MVU estimator $\hat{p}_{\mathrm{mvu}}$ into
(\ref{EQ:Limits-h-chart}) and (\ref{EQ:Limits-g-chart}).
However, like the ML estimator, the MVU estimator has no invariance property.
Thus, in this case, the resulting limits can not be regarded as the MVU-based limits.

%======================================
\section{Concluding remarks}  \label{Section:05}
%======================================
We have revisited the $g$ and $h$ control charts with proper ML and MVU estimators.
We have shown that the MVU estimator has been inappropriately used in the quality engineering literature
and thus provided the correct MVU estimator along with various statistical properties such as
their theoretical first and second moments which are explicitly expressed as the Gauss hypergeometric function.
Furthermore, based on the new estimators developed in this note, 
we provided how to construct the ML-based and MVU-based $h$ and $g$ control charts with unbalanced samples.

Finally, it is worth noting that we have developed the \texttt{rQCC} R package
\citep{Park/Wang:2020b} to construct various control charts. In ongoing work, 
we plan to add these control charts in the next update so that practitioners can use our results more easily.

%%===============================
\section*{Acknowledgment}
%%===============================
This research was supported by the National Research Foundation of Korea (NRF) grant
(NRF-2017R1A2B4004169) and the BK21-Plus Program (Major in Industrial Data Science and Engineering)
funded by the Korea government.

%%===============================================
\bibliographystyle{chicago}
\bibliography{REFmw92}
%%===============================================

%%===============================================
\end{document}